\documentclass[runningheads,envcountsame,envcountsect,11pt]{llncs}
\usepackage[a4paper,margin=0.85in]{geometry}

\usepackage{graphicx}

\usepackage{hyperref,xcolor}

\usepackage{amsmath,amssymb}

\spnewtheorem{observation}[theorem]{Observation}{\bfseries}{\itshape}
\spnewtheorem{myclaim}[theorem]{Claim}{\bfseries}{\itshape}

\newcommand{\fig}{./fig}

\usepackage{color}
\definecolor{darkred}{rgb}{0.8,0,0}

\usepackage{tabularx,environ}

\makeatletter
\newcommand{\problemtitle}[1]{\gdef\@problemtitle{#1}}
\newcommand{\probleminput}[1]{\gdef\@probleminput{#1}}
\newcommand{\problemquestion}[1]{\gdef\@problemquestion{#1}}
\NewEnviron{myproblem}{
  \problemtitle{}\probleminput{}\problemquestion{}
  \BODY
  \par\addvspace{.5\baselineskip}
  \noindent \normalsize
  \begin{tabularx}{\textwidth}{@{\hspace{\parindent}} l X c}
    \multicolumn{2}{@{\hspace{\parindent}}l}{\normalsize\@problemtitle} \\
    \normalsize \textbf{Input:} & \normalsize \@probleminput \\
    \normalsize \textbf{Question:} & \normalsize \@problemquestion
  \end{tabularx}
  \par\addvspace{.5\baselineskip}
}
\makeatother

\newcommand{\figref}[1]{\figurename~\ref{#1}}

\newcommand{\bsat}{3-\textsc{Sat}($2, 1$)}

\newcommand{\vi}[1]{\mathsf{vi}(#1)}

\renewcommand{\orcidID}[1]{} 

\newcommand{\col}{\mathsf{col}}

\begin{document}

\title{Subgraph Isomorphism on Graph Classes that Exclude a Substructure\thanks{%
Partially supported by
NETWORKS (the Networks project, funded by the Netherlands Organization for Scientific Research NWO),
the ELC project (the project Exploring the Limits of Computation, funded by MEXT),
JSPS/MEXT KAKENHI grant numbers JP24106004, JP18K11168, JP18K11169, JP18H04091, JP18H06469, JP15K00009,
JST CREST Grant Number JPMJCR1402, and Kayamori Foundation of Informational Science Advancement.
The authors thank 
Momoko Hayamizu, Kenji Kashiwabara, Hirotaka Ono, Ryuhei Uehara, and Koichi Yamazaki for helpful discussions.
The authors are grateful to the anonymous reviewer of an earlier version of this paper
who pointed out a gap in a proof.
}}
%

\author{Hans L. Bodlaender\inst{1}\orcidID{0000-0002-9297-3330} \and
Tesshu Hanaka\inst{2}\orcidID{0000-0001-6943-856X} \and
Yasuaki Kobayashi\inst{3}\orcidID{} \and
Yusuke Kobayashi\inst{3}\orcidID{} \and \\ 
Yoshio Okamoto\inst{4,5}\orcidID{0000-0002-9826-7074} \and
Yota Otachi\inst{6}\orcidID{0000-0002-0087-853X} \and
Tom C. van der Zanden\inst{1}\orcidID{0000-0003-3080-3210}
}
\authorrunning{Bodlaender et al.}
%

\institute{
Utrecht University, Utrecht, The Netherlands\\
\email{\{H.L.Bodlaender,T.C.vanderZanden\}@uu.nl}
\and
Chuo University, Bunkyo-ku, Tokyo, Japan\\
\email{hanaka.91t@g.chuo-u.ac.jp}
\and
Kyoto University, Kyoto, Japan\\
\email{yusuke@kurims.kyoto-u.ac.jp}, \email{kobayashi@iip.ist.i.kyoto-u.ac.jp}
\and
The University of Electro-Communications, Chofu, Tokyo, Japan\\
\email{okamotoy@uec.ac.jp}
\and
RIKEN Center for Advanced Intelligence Project, Tokyo, Japan
\and
Kumamoto University, Kumamoto, 860-8555, Japan\\
\email{otachi@cs.kumamoto-u.ac.jp}
}

\maketitle

\begin{abstract}
We study \textsc{Subgraph Isomorphism} on graph classes defined by a fixed forbidden graph.
Although there are several ways for forbidding a graph,
we observe that it is reasonable to focus on the minor relation
since other well-known relations lead to either trivial or equivalent problems.
When the forbidden minor is connected, we present a near dichotomy of the complexity of \textsc{Subgraph Isomorphism}
with respect to the forbidden minor, where the only unsettled case is $P_{5}$, the path of five vertices.
We then also consider the general case of possibly disconnected forbidden minors.
We show fixed-parameter tractable cases and 
randomized XP-time solvable cases parameterized by the size of the forbidden minor $H$.
We also show that by slightly generalizing the tractable cases, the problem becomes NP-complete.
All unsettle cases are equivalent to $P_{5}$ or the disjoint union of two $P_{5}$'s.
As a byproduct, we show that \textsc{Subgraph Isomorphism}
is fixed-parameter tractable parameterized by vertex integrity.
Using similar techniques, we also observe that 
\textsc{Subgraph Isomorphism} is fixed-parameter tractable parameterized by neighborhood diversity.

\keywords{Subgraph isomorphism \and minor-free graphs \and parameterized complexity}
\end{abstract}


\section{Introduction}
\label{sec:intro}
Let $Q$ and $G$ be graphs.
A \emph{subgraph isomorphism} $\eta$ is an injection from $V(Q)$ to $V(G)$
that preserves the adjacency in $Q$; that is, if $\{u,v\} \in E(Q)$, then $\{\eta(u),\eta(v)\} \in E(G)$.
We say that $Q$ is \emph{subgraph-isomorphic} to $G$ if there is a subgraph isomorphism from $Q$ to $G$,
and write $Q \preceq G$.
In this paper, we study the following problem of deciding the existence of a subgraph isomorphism.
\begin{myproblem}
  \problemtitle{\textsc{Subgraph Isomorphism}}
  \probleminput{Two graphs $G$ (the \emph{host} graph) and $Q$ (the \emph{pattern} graph).}
  \problemquestion{$Q \preceq G$?}
\end{myproblem}

The problem \textsc{Subgraph Isomorphism} is one of the most general and fundamental graph problems
and generalizes many other graph problems such as 
\textsc{Graph Isomorphism}, \textsc{Clique}, \textsc{Hamiltonian Path/Cycle}, and \textsc{Bandwidth}.
Obviously, \textsc{Subgraph Isomorphism} is NP-complete in general.
When both host and pattern graphs are restricted to be in a graph class $\mathcal{C}$,
we call the problem \emph{\textsc{Subgraph Isomorphism} on $\mathcal{C}$}.
By slightly modifying known reductions in \cite{GareyJ79,Damaschke90},
one can easily show that the problem is hard even for very restricted graph classes.
Recall that a \emph{linear forest} is the disjoint union of paths
and a \emph{cluster graph} is the disjoint union of complete graphs.
We can show the following hardness of \textsc{Subgraph Isomorphism} by a simple reduction from \textsc{3-Partition}~\cite{GareyJ79}.
See Proposition~\ref{prop:linfor-cluster}.

Since most of the well-studied graph classes contain all linear forests or all cluster graphs,
it is often hopeless to have a polynomial-time algorithm for an interesting graph class.
This is sometimes true even if we further assume that the graphs are connected~\cite{KijimaOSU12,KonagayaOU16}.
On the other hand, it is polynomial-time solvable for trees~\cite{Matula78}.
This result was first generalized for 2-connected outerplanar graphs~\cite{Lingas89},
and finally for $k$-connected partial $k$-trees~\cite{MatousekT92,GuptaN96} (where the running time is XP parameterized by $k$).
In \cite{MatousekT92}, a polynomial-time algorithm for partial $k$-trees of bounded maximum degree is presented as well,
which is later generalized to partial $k$-trees of log-bounded fragmentation~\cite{HajiaghayiN07}.
When the pattern graph has bounded treewidth,
the celebrated color-coding technique~\cite{AlonYZ95} gives a fixed-parameter algorithm parameterized by the size of the pattern graph.
It is also known that for chain graphs, co-chain graphs, and threshold graphs,
\textsc{Subgraph Isomorphism} is polynomial-time solvable~\cite{KijimaOSU12,KonagayaOU16,KiyomiO16}.
In the case where only the pattern graph has to be in a restricted graph class that is closed under vertex deletions,
a complexity dichotomy with respect to the graph class is known~\cite{JansenM15}.

Because of its unavoidable hardness in the general case,
it is often assumed that the pattern graph is small.
In such a setting, we can study the parameterized complexity\footnote{%
We assume that the readers are familiar with the concept of parameterized complexity.
See e.g.\ \cite{CyganFKLMPPS15} for basic definitions omitted here.
}
of \textsc{Subgraph Isomorphism} parameterized by the size of the pattern graph.
Unfortunately, the W[1]-completeness of \textsc{Clique}~\cite{DowneyF95a}
implies that this parameterization does not help in general.
Indeed, the existence of a $2^{o(n \log n)}$-time algorithm for \textsc{Subgraph Isomorphism}
 is ruled out assuming the Exponential Time Hypothesis,
where $n$ is the total number of vertices~\cite{CyganFGKMPS17}.
So we need further restrictions on the considered graph classes even in the parameterized setting.
For planar graphs, it is known to be fixed-parameter tractable~\cite{Eppstein99,Dorn10}.
This result is later generalized to graphs of bounded genus~\cite{Bonsma12}.
For several graph parameters, the parameterized complexity of \textsc{Subgraph Isomorphism}
parameterized by combinations of them is determined in \cite{MarxP14}.
In \cite{BodlaenderNZ16}, it is shown that when the pattern graph excludes a fixed graph as a minor,
the problem is fixed-parameter tractable parameterized by treewidth and the size of the pattern graph.
The result in \cite{BodlaenderNZ16} implies also that \textsc{Subgraph Isomorphism} 
can be solved in subexponential time when the host graph also excludes a fixed graph as a minor.


\subsection{Our results}

As mentioned above, the research on \textsc{Subgraph Isomorphism} has been done mostly
when the size of the pattern graph is considered as a parameter.
However, in this paper, we are going to study the general case where the pattern graph can be as large as the host graph.

We denote the path of $n$ vertices by $P_{n}$,
the complete graph of $n$ vertices by $K_{n}$,
and the star with $\ell$ leaves by $K_{1,\ell}$.
The disjoint union of graphs $X$ and $Y$ is denoted by $X \cup Y$
and the disjoint union of $k$ copies of a graph $Z$ is denoted by $k Z$.
The complement of a graph $X$ is denoted by $\overline{X}$.

We believe the following fact is folklore but give a proof to be self-contained.
\begin{proposition}
\label{prop:linfor-cluster}
\textsc{Subgraph Isomorphism} is NP-complete on linear forests and on cluster graphs
even if the input graphs have the same number of vertices.
\end{proposition}
\begin{proof}
Since the problem is clearly in NP, we show the NP-hardness.
Recall that \textsc{3-Partition} is the following problem:
the input is $3m$ positive integers $a_{1}, \dots, a_{3m}$ with $\sum_{1 \le i \le 3m} a_{i} = m B$ such that $B/4 < a_{i} < B/2$ for all $i$.
The task is to decide whether there is a partition of $a_{1}, \dots, a_{3m}$ into $m$ triplets such that the members of each triplet sum up to $B$.
The problem is known to be strongly NP-complete;
that is, it is NP-complete even if a polynomial in $m$ upper bounds all $a_{i}$~\cite{GareyJ79}.

To show the hardness on linear forests,
we set the host graph $G$ to be $m P_{B}$
and the pattern graph $Q$ to be $P_{a_{1}} \cup \dots \cup P_{a_{3m}}$.
Similarly, to show the hardness on cluster graphs,
we set $G = m K_{B}$ and $Q = K_{a_{1}} \cup \dots \cup K_{a_{3m}}$.
It is straightforward to show that $Q \preceq G$ if and only if
the corresponding instance of \textsc{3-Partition} is a yes-instance.
\qed
\end{proof}

Our first observation is that forbidding a graph as an induced substructure
(an induced subgraph, an induced topological minor, or an induced minor)
does not help for making \textsc{Subgraph Isomorphism} tractable unless we make the graph class trivial.
This can be done just by combining some easy observations and known results.

\begin{observation}
\label{obs:induced}
Let $\mathcal{C}$ be the graph class that forbids a fixed graph $H$ as
either an induced subgraph, an induced topological minor, or an induced minor.
Then, \textsc{Subgraph Isomorphism} on $\mathcal{C}$ is polynomial-time solvable
if $H$ has at most two vertices; otherwise, it is NP-complete.
\end{observation}
\begin{proof}
We assume that $H$ is a linear forest since otherwise it is NP-complete by Lemma~\ref{lem:linear-forest}.
Hence, we may assume that $H$ is forbidden as an induced subgraph as it is equivalent to the other cases.

If $H = 3K_{1}$, then \textsc{Clique} is NP-complete on $\mathcal{C}$
because \textsc{Independent Set} is NP-complete on triangle-free graphs~\cite{Poljak74}.
If $H = P_{3}$, then $\mathcal{C}$ is the class of cluster graphs and thus 
\textsc{Subgraph Isomorphism} on $\mathcal{C}$ is NP-complete by Proposition~\ref{prop:linfor-cluster}.
If $H = K_{2} \cup K_{1} = \overline{P_{3}}$, 
then $\mathcal{C}$ is the class of co-cluster graphs (or complete multi-partite graphs).
It is known that if $|V(G)| = |V(Q)|$, then $Q \preceq G$ if and only if $\overline{G} \preceq \overline{Q}$~\cite{KijimaOSU12}.
Thus, by Proposition~\ref{prop:linfor-cluster}, this case is NP-complete.

In the remaining cases, we can assume that $H$ has order at most 2
since all other cases are NP-complete by the discussion above.
For these cases, the allowed graphs are either edgeless or complete,
and thus \textsc{Subgraph Isomorphism} is trivially polynomial-time solvable.
\qed
\end{proof}

Our main contribution in this paper is the following pair of results on 
\textsc{Subgraph Isomorphism} on graph classes forbidding a fixed graph as a substructure.

\begin{theorem}
\label{thm:connected}
Let $\mathcal{C}$ be the graph class that forbids a fixed connected graph $H \ne P_{5}$ as
either a subgraph, a topological minor, or a minor.
Then, \textsc{Subgraph Isomorphism} on $\mathcal{C}$ is polynomial-time solvable
if $H$ is a subgraph of $P_{4}$; otherwise, it is NP-complete.
\end{theorem}

\begin{theorem}
\label{thm:disconnected}
Let $\mathcal{C}$ be the graph class that forbids a fixed (not necessarily connected) graph $H$ as
either a subgraph, a topological minor, or a minor.
Then, \textsc{Subgraph Isomorphism} on $\mathcal{C}$ is 
\begin{itemize}
  \item fixed-parameter tractable parameterized by the order of $H$
  if $H$ is a linear forest 
  such that at most one component is of order 4
  and all other components are of order at most 3;

  \item randomized XP-time solvable parameterized by the order of $H$
  if $H$ is a linear forest 
  such that each component is of order at most 4;

  \item NP-complete if either
  $H$ is not a linear forest,
  $H$ contains a component with six or more vertices, or
  $H$ contains three components with five vertices.
\end{itemize}
All other cases are randomized polynomial-time reducible
to the case where $H$ is $P_{5}$ or $2P_{5}$.
\end{theorem}

We prove Theorem~\ref{thm:connected} in Section~\ref{sec:connected} and
Theorem~\ref{thm:disconnected} in Section~\ref{sec:disconnected}.

\section{Preliminaries and basic observations}
\label{sec:pre}

A graph $Q$ is a \emph{minor} of $G$ if
$Q$ can be obtained from $G$ by removing vertices, removing edges,
and contracting edges, where contracting an edge $\{u,v\}$
means adding a new vertex $w_{u,v}$,
making the neighbors of $u$ and $v$ adjacent to $w_{u,v}$,
and removing $u$ and $v$.
A graph $Q$ is a \emph{topological minor} of $G$ if
$Q$ can be obtained by removing vertices, removing edges,
and contracting edges, where contraction of an edge is allowed
if one of the endpoints of the edge is of degree 2.
A graph $Q$ is a \emph{subgraph} of $G$ if
$Q$ can be obtained by removing vertices and edges.
If we cannot remove edges but can do the other modifications as before,
then we get the induced variants
\emph{induced minor}, \emph{induced topological minor}, and \emph{induced subgraph}.

Recall that a graph is a linear forest if it is the disjoint union of paths.
In other words, a graph is a linear forest if and only if 
it does not contain a cycle nor a vertex of degree at least 3.
Observe that in all graph containment relations mentioned above,
if we do not forbid any linear forest from a graph class, then the class includes all linear forests.
Thus, by Proposition~\ref{prop:linfor-cluster}, we have the following lemma.
\begin{lemma}
\label{lem:linear-forest}
If $H$ is not a linear forest, then \textsc{Subgraph Isomorphism} is NP-complete
for graphs that do not contain $H$ as
a minor, a topological minor, a subgraph,
an induced minor, an induced topological minor, or an induced subgraph.
\end{lemma}

\subsection{Graphs forbidding a short path as a minor}

By the discussion above, we can focus on a graph class forbidding a linear forest as a minor
(or equivalently as a topological minor or a subgraph).
We here characterize graph classes forbidding a short path as a minor.

\begin{lemma}
\label{lem:P3-minor-free_characterization}
A connected $P_{3}$-minor free graph is isomorphic to $K_{1}$ or $K_{2}$.
\end{lemma}
\begin{proof}
If a connected graph is not complete, then there is a path between nonadjacent vertices.
This path contains $P_{3}$ as a subgraph. A complete graph with more than two vertices
contains $P_{3}$ as a subgraph.
\qed
\end{proof}

\begin{lemma}
\label{lem:P4-minor-free_characterization}
A connected $P_{4}$-minor free graph is isomorphic to $K_{1}$, $K_{3}$, or $K_{1,s}$ for some $s \ge 1$.
\end{lemma}
\begin{proof}
Let $H$ be a connected $P_{4}$-minor free graph.
If $H$ is not a tree, it has a cycle $C$ as a subgraph.
As $H$ is $P_{4}$-minor free, this cycle $C$ has length 3.
If $H$ contains a vertex $v$ that is not in $C$ but has a neighbor on $C$,
then the vertices in $C$ together with $v$ induce a subgraph of $H$ that contains $P_{4}$ as a subgraph.
The connectivity of $H$ implies that $H = C = K_{3}$.

Now assume that $H$ is a tree with two or more vertices.
If $H$ has no universal vertex, then there are two edges $e_{1}, e_{2} \in E(H)$ that do not share any endpoint.
The edges $e_{1}$ and $e_{2}$ with the unique path connecting them form a path of at least four vertices.
Thus $H$ has a universal vertex, and hence it is a star.
\qed
\end{proof}

\section{Forbidding a connected graph as a minor}
\label{sec:connected}

Here we first show that \textsc{Subgraph Isomorphism} on $P_{k}$-minor free graphs is linear-time solvable if $k \le 4$.
Note that $P_{k}$-minor free graphs include all $P_{k'}$-minor free graphs if $k' \le k$.

The following result can be easily obtained from Lemma~\ref{lem:P4-minor-free_characterization}.
\begin{lemma}
\label{lem:p4}
\textsc{Subgraph Isomorphism} on $P_{4}$-minor free graphs is linear-time solvable.
\end{lemma}
\begin{proof}
Let $G$ be the host graph and $Q$ be the pattern graph.
We assume that $|V(Q)| \le |V(G)|$ since otherwise $Q \not\preceq G$.
By Lemma~\ref{lem:P4-minor-free_characterization},
each component in both graphs is either an isolated vertex $K_{1}$, a triangle $K_{3}$, or a star $K_{1,s}$ for some $s \ge 1$.

We first remove isolated vertices.
Let $Q'$ and $G'$ be the graphs obtained from $Q$ and $G$, respectively, by removing all isolated vertices.
Since $|V(Q)| \le |V(G)|$, $Q \preceq G$ if and only if $Q' \preceq G$.
Also, since isolated vertices in $G$ cannot be used to embed any vertex of $Q'$,
$Q' \preceq G$ if and only if $Q' \preceq G'$.
In the following, we assume that $Q$ and $G$ do not have isolated vertices.

We next get rid of the triangles.
A triangle $K_{3}$ in $Q$ has to be matched to a triangle in $G$.
Therefore, if $Q$ contains $t$ triangles,
we can remove $t$ triangles from each of $G$ and $Q$,
and obtain an equivalent instance.
(If $G$ does not contain $t$ triangles, then we can immediately find that $Q \not\preceq G$.)
Since $Q$ contains no triangle anymore,
all triangles $K_{3}$ in $G$ can be replaced with the same number of $K_{1,2}$'s.

Now we have only stars $K_{1,s}$ with $s \ge 1$ in both graphs.
The rest of the problem can be solved by greedily matching a maximum star in $Q$ to a maximum star in $G$.

The preprocessing phase can be done in linear time.
The matching phase can be done in linear time as well since
we just need to bucket sort the component sizes in each graph and compare them.
\qed
\end{proof}


The following theorem implies that \textsc{Subgraph Isomorphism} on $P_{k}$-minor free graphs is NP-complete
for every $k \ge 6$.

\begin{theorem}
\label{thm:p6}
\textsc{Subgraph Isomorphism} is NP-complete when
the host graph is a forest without paths of length 6 and
the pattern is a collection of stars.
\end{theorem}

\begin{proof}
The problem clearly is in NP\@.
To show hardness, we reduce from \textsc{Exact 3-Cover}~\cite{GareyJ79}:

\begin{myproblem}
  \problemtitle{\textsc{Exact 3-Cover}}
  \probleminput{Collection $\mathcal{C}$ of subsets of a set $U$ such that each $c \in \mathcal{C}$ has size~3.}
  \problemquestion{Is there a subcollection $\mathcal{C}' \subseteq \mathcal{C}$
    such that $\bigcup_{C \in \mathcal{C}'} C = U$ and $|\mathcal{C'}| = |U|/3$?}
\end{myproblem}

Suppose we have an instance $(\mathcal{C}, U)$ of \textsc{Exact 3-Cover} given, where $U = \{u_0, \ldots, u_{n-1}\}$.
From $(\mathcal{C}, U)$, we construct the host graph $G$ and the pattern $Q$.

The host $G$ consists of the disjoint union of $|\mathcal{C}|$  trees as follows (see \figref{fig:p6mfreeG}).
For each set $C \in \mathcal{C}$, we take a tree in $G$ as follows.
Take a star $K_{1, 4n+6}$. For each $u_{i} \in C$, do the following: 
take one of the leaves of the star, and add $n+i$ pendant vertices to it.
Take another leaf of the star, and add $3n-i$ pendant vertices to it.
I.e., if $C = \{u_{i},u_{j},u_{k}\}$, then the corresponding tree has seven
vertices of degree more than 1: one vertex with degree $4n+6$, which is also
adjacent to each of the other six non-leaf vertices; the non-leaf vertices have
degree $n+i+1$, $3n-i+1$, $n+j+1$, $3n-j+1$, $n+k+1$, and $3n-k+1$.
Call the vertex of degree $4n+6$ the \emph{central} vertex of the component of $C$.

\begin{figure}[tb]
  \centering
  \includegraphics[scale=.6]{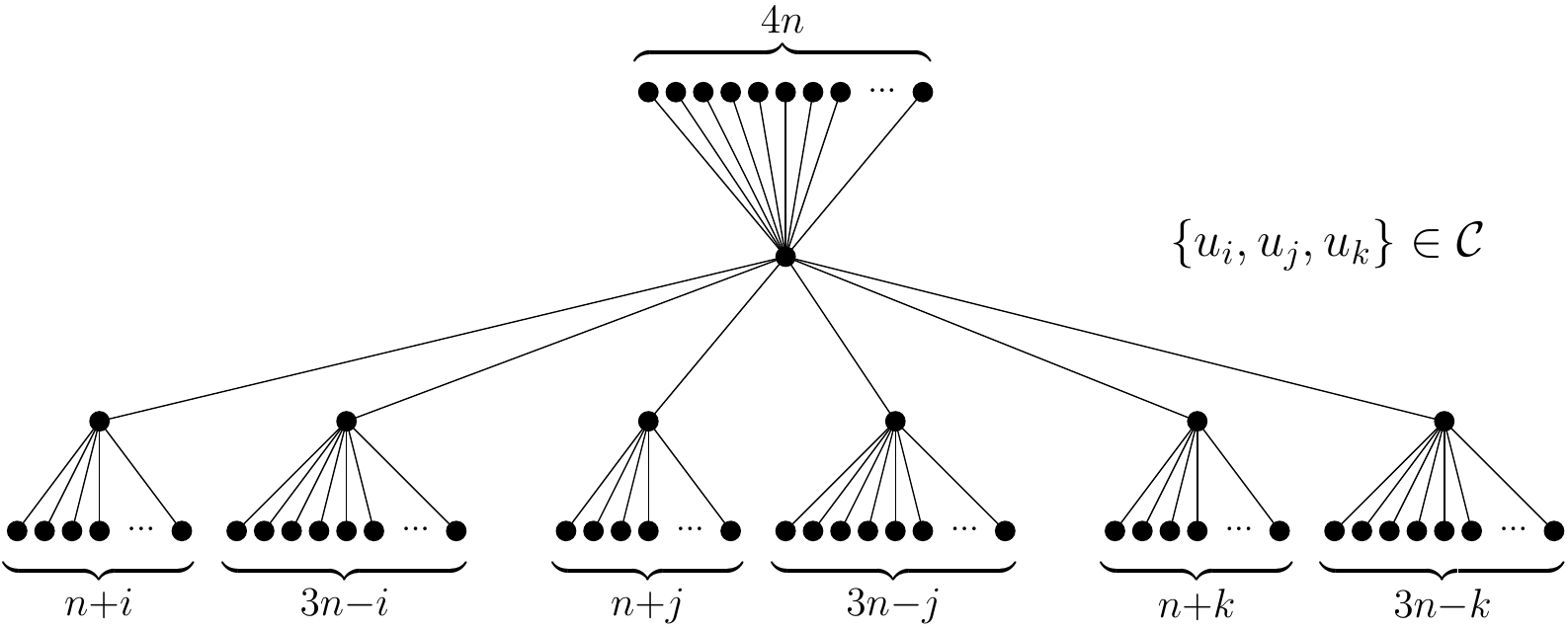}
  \caption{The tree in $G$ corresponding to $\{u_{i},u_{j},u_{k}\} \in \mathcal{C}$.}
  \label{fig:p6mfreeG}
\end{figure}

The pattern graph $Q$ consists of a number of stars (see \figref{fig:p6mfreeQ}):
\begin{itemize}
\item We have $n/3$ stars $K_{1,4n}$.
\item We have $|{\cal C}| - n/3$ stars $K_{1,4n+6}$.
\item For each $i \in \{0,\dots,n-1\}$, we have stars $K_{1,n+i}$ and $K_{1,3n-i}$.
Call these the \emph{element stars}.
\end{itemize}

\begin{figure}[tb]
  \centering
  \includegraphics[scale=.6]{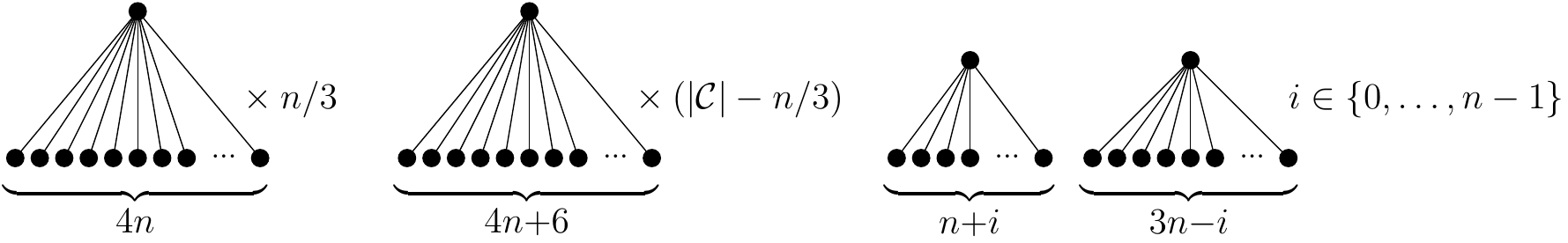}
  \caption{The pattern graph $Q$.}
  \label{fig:p6mfreeQ}
\end{figure}

From $(\mathcal{C}, U)$, $G$ and $Q$ can be constructed in polynomial time.
Now we show that $Q \preceq G$ if and only if $(\mathcal{C}, U)$ is a yes-instance of \textsc{Exact 3-Cover}.
We assume that $n > 6$ in the following.

\paragraph{The if direction:}
Suppose that the \textsc{Exact 3-Cover} instance $(\mathcal{C}, U)$ has a solution $\mathcal{C}' \subseteq \mathcal{C}$. 

We map each $K_{1,4n+6}$ of $Q$ into a component $M$ of $G$ corresponding to a set $D \notin \mathcal{C}'$. 
The center of $K_{1,4n+6}$ is mapped to the central vertex of $M$ and all leaves to its neighbors. 
The other vertices $T$ are isolated and not used.

Embed each $K_{1,4n}$ of $Q$ into a component $L$ of $G$ corresponding to a set $C \in \mathcal{C}'$, 
mapping the center of $K_{1,4n}$ to the central vertex of $L$, and the leaves of $K_{1,4n}$ to leaves neighboring the central vertex of $L$.
After we have done so, we left in this component six stars:
if $C = \{u_{i}, u_{j}, u_{k}\}$, then the vertices in $L$  that we did not yet use form 
stars $K_{1,n+i}$, $K_{1,3n-i}$, $K_{1,n+j}$, $K_{1,3n-j}$, $K_{1,n+k}$, $K_{1,3n-k}$.
We thus can embed the element stars corresponding to $u_{i}$, $u_{j}$, and $u_{k}$
in these stars, and have embedded the entire pattern in the host graph since $\mathcal{C}'$ is a cover of $U$.

\paragraph{The only if direction:}
Suppose that $Q \preceq G$. 
Note that both $Q$ and $G$ have exactly $|\mathcal{C}|$ vertices of degree at least $4n$.
Thus it follows that each vertex of degree at least $4n$ in $Q$ must be mapped to a central vertex of a component in $G$.
We can see that one of the following two cases must hold for the components in the host graph $G$.

\textbf{Case 1: A star $K_{1,4n+6}$ is embedded in the component.}
This ``uses up'' the central vertex and all its neighbors.
The only vertices in the component that are not in the image of the star $K_{1,4n+6}$ are leaves with its neighbor being used: 
these isolated vertices thus cannot be used for embedding any other stars.
So all element stars must be embedded in components for which Case 2 holds.

\textbf{Case 2: A star $K_{1,4n}$ is embedded in the component.}
At this point, note that the total number of vertices of element stars in $Q$ equals $4 n^{2} + 2n$:
each of the $n$ elements has in total $4n$ leaves and two high degree vertices in its element stars.
Also, the total number of vertices not used by the stars $K_{1,4n}$ in the Case 2-components equals $4n^{2}+2n$: 
we have $n/3$ components of Case 2 in $G$ and each has $16n+7$ vertices of which $4n+1$ are used for embedding the star $K_{1,4n}$. 
Thus, each vertex in a Case 2-component $M$ must be used for embedding a vertex.
This is only possible if we embed in $M$ the element stars of the elements in the set corresponding to $M$.

So, let $\cal C'$ be the sets whose component is of Case 2, i.e., where we embedded
a $K_{1,4n}$ in its component. This subcollection $\cal C'$ is a solution for \textsc{Exact 3-Cover}:
for each element $u_i$, its element stars are embedded in a component that corresponds
to a set $C$ that contains $u_i$, and by the argument above $C \in {\cal C}'$.
\qed
\end{proof}

By Lemma~\ref{lem:linear-forest}, if a connected graph $H$ is not a path, then 
\textsc{Subgraph Isomorphism} on $H$-minor free graphs is NP-complete.
Assume that $H$ is a path $P_{k}$. If $k \ge 6$, then by Theorem~\ref{thm:p6} the problem is NP-complete.
If $k \le 4$, then by Lemma~\ref{lem:p4} the problem can be solved in polynomial time.
This completes the proof of Theorem~\ref{thm:connected}.


\section{Forbidding a disconnected graph as a minor}
\label{sec:disconnected}

In this section, we study the more general cases where the forbidden minor $H$ is not necessarily connected.
By Lemma~\ref{lem:linear-forest}, we can focus on linear forests $H$.
We already know, by Theorem~\ref{thm:p6}, if $H$ contains a component with six or more vertices the problem becomes NP-complete.
Thus in the following we consider the case where the components of $H$ have five or less vertices.

Using the results in this section, we can prove Theorem~\ref{thm:disconnected}.
Corollary~\ref{cor:kP3} implies the positive case of Theorem~\ref{thm:disconnected}.
Theorems~\ref{thm:p6} and \ref{thm:4P5}
together with Lemma~\ref{lem:linear-forest} imply the negative cases.
The discussion on the missing cases will be in Section~\ref{ssec:missing-cases}.

\subsection{Subgraph isomorphism on $(P_{4} \cup k P_{3})$-minor free graphs}

We show that \textsc{Subgraph Isomorphism} on $(P_{4} \cup k P_{3})$-minor free graphs is fixed-parameter tractable when parameterized by $k$. 
To this end, we present an algorithm that is parameterized by the vertex integrity, which we think is of independent interest.
The \emph{vertex integrity}~\cite{BarefootES87} of a graph is the minimum integer $k$ such that
there is a vertex set $S \subseteq V$ such that $|S| \le k$ and the maximum order of the components of $G - S$ is at most $k - |S|$.
We call such $S$ a \emph{$\vi{k}$ set} of $G$.
Note that the property of having vertex integrity at most $k$ is closed under the subgraph relation.

This subsection is devoted to the proof of the following theorem.
\begin{theorem}
\label{thm:vertex-integrity}
\textsc{Subgraph Isomorphism} on graphs of vertex integrity at most $k$
is fixed-parameter tractable when parameterized by $k$.
\end{theorem}

By combining Theorem~\ref{thm:vertex-integrity}, Lemma~\ref{lem:p4},
and the fact that $k P_{3}$-minor free graphs have vertex integrity at most $3k-1$,
we can prove the following. 
\begin{corollary}
\label{cor:kP3}
\textsc{Subgraph Isomorphism} on $(P_{4} \cup k P_{3})$-minor free graphs is fixed-parameter tractable when parameterized by $k$. 
\end{corollary}
\begin{proof}
Let $G$ be the host graph and $Q$ be the pattern graph.
We first check whether the input graphs are $P_{4}$-minor free.
This can be done in polynomial time since we just need to check
the existence of a $P_{4}$ subgraph.
If $G$ is $P_{4}$-minor free but $Q$ is not, then $Q\not\preceq G$.
If both are $P_{4}$-minor free,
then the problem can be solved in polynomial time by Lemma~\ref{lem:p4}.
Hence, in the following, we assume that $G$ has a $P_{4}$ minor, and thus a subgraph $R$ isomorphic to $P_{4}$.
Since $G$ is $(P_{4} \cup k P_{3})$-minor free, $G - V(R)$ is $k P_{3}$-minor free.

Observe that a $k P_{3}$-minor free graph has vertex integrity at most $3k-1$:
by repeatedly removing vertices that form a $P_{3}$ subgraph at most $k-1$ times, we can make the graph $P_{3}$-minor free,
which has the maximum component order at most $2$ by Lemma~\ref{lem:P3-minor-free_characterization}.
Therefore, $G$ itself has vertex integrity at most $3k-1 + |V(R)| = 3k + 3$.

We now check whether $Q$ has vertex integrity at most $3k+3$, which can be done in FPT time parameterized by $k$~\cite{DrangeDH16}.
If this is not the case, then $Q \not\preceq G$.
If $Q$ has vertex integrity at most $3k+3$, then we can apply Theorem~\ref{thm:vertex-integrity}.
\qed
\end{proof}

To prove Theorem~\ref{thm:vertex-integrity},
we start with the following simple fact.
\begin{lemma}
\label{lem:vi(k)-set_subset}
Let $\eta$ be a subgraph isomorphism from $Q$ to $G$.
For every $\vi{k}$ set $T$ of $G$, there exists a minimal $\vi{k}$ set $S$ of $Q$
such that $\eta(S) \subseteq T$.
\end{lemma}
\begin{proof}
Let $G' = G[\eta(V(Q))]$ and $T' = T \cap \eta(V(Q))$.
The set $T'$ is a $\vi{k}$ set of $G'$.
Let $S' = \eta^{-1}(T')$. 
Since $\eta$ restricted to $V(Q) - S'$ is a subgraph isomorphism from $Q-S'$ to $G' - T'$,
the maximum component order of $Q-S'$ cannot be larger than that of $G' - T$,
and thus $S'$ is a $\vi{k}$ set of $Q$.
Now every minimal $\vi{k}$ set $S \subseteq S'$ of $Q$ satisfies that $\eta(S) \subseteq T' \subseteq T$.
\qed
\end{proof}

Our algorithm assumes that there is a subgraph isomorphism $\eta$ from $Q$ to $G$
and proceeds as follows:
\begin{enumerate}
  \item find a $\vi{k}$ set $T$ of $G$;

  \item guess a minimal $\vi{k}$ set $S$ of $Q$ such that $\eta(S) \subseteq T$;

  \item guess the bijection between $S$ and $R := \eta(S)$;

  \item guess a subset $F \subseteq E(G - R)$ of the edges ``unused'' by $\eta$ such that $R$ is a $\vi{k}$ set of $G - F$;
  
  \item solve the problem of deciding the extendability of the guessed parts
  as the feasibility problem of an integer linear program with a bounded number of variables.
\end{enumerate}

\begin{proof}
[Theorem~\ref{thm:vertex-integrity}]
Let $G$ and $Q$ be graphs of vertex integrity at most $k$.
Our task is to find a subgraph isomorphism $\eta$ from $Q$ to $G$ in FPT time parameterized by $k$.

We first find a $\vi{k}$ set $T$ of $G$
and then guess a minimal $\vi{k}$ set $S$ of $Q$ such that $\eta(S) \subseteq T$ for some subgraph isomorphism $\eta$ from $Q$ to $G$.
By Lemma~\ref{lem:vi(k)-set_subset}, such a set $S$ exists if $\eta$ exists.
Finding $T$ can be done in $O(k^{k+1} n)$ time~\cite{DrangeDH16}, where $n = |V(G)|$.
To guess $S$, it suffices to list all minimal $\vi{k}$ set $S$ of $Q$.
The same algorithm in~\cite{DrangeDH16} can be used again:
it lists all $O(k^{k})$ candidates by branching on $k+1$ vertices that induce a connected subgraph.

We then guess the subset $R$ of $T$ such that $\eta(S) = R$.
We also guess for each $s \in S$, the image $\eta(s) \in R$.
That is, we guess an injection from $S$ to $T$.
The number of such injections is $\binom{|T|}{|S|} \cdot |S|! \le k!$.
If there is an edge $\{u,v\} \in E(Q[S])$ such that $\{\eta(u),\eta(v)\} \notin E(G[R])$,
then we reject this guess. Otherwise, we try to further extend $\eta$.

Observe that $R$ is not necessarily a $\vi{k}$ set of $G$.
In the following, we guess ``unnecessary'' edges in $G-R$.
That is, we guess a subset $F$ of the edges that are not used by $\eta$ as images of any edges in $Q$.
Furthermore, we select $F$ so that $R$ is a $\vi{k}$ set of $G - F$.
Such $F$ exists because $\eta$ embeds $Q-S$ (and no other things) into $G-R$.

\paragraph{Guessing $F$:}

We now show that the number of candidates of $F$ that we need to consider is bounded by some function in $k$.
We partition $F$ into three sets $F_{1} = F \cap E(G[T-R])$, $F_{2} = F \cap E(V(G)-T, T-R)$,
and $F_{3} = F \cap E(G-T)$ and then count the numbers of candidates separately.

\textbf{Guessing $F_{1}$:}
For $F_{1}$, we just use all $2^{|E(G[T-R])|} < 2^{k^{2}}$ subsets of $E(G[T-R])$ as candidates.
If $R$ is not a $\vi{k}$ set of $G[T] - F_{1}$, we reject this $F_{1}$.

\textbf{Guessing $F_{2}$:}
Since we are finding $F$ such that $R$ is a $\vi{k}$ set of $G - F$,
each vertex in $T-R$ has less than $k$ edges to $V(G) - T$ in $G - F$.
Thus fewer than $k^{2}$ components of $V(G)-T$ have edges to $T-R$ in $G - F$.
We guess such components $\mathcal{C}$.

Observe that each component in $V(G) - T$ is of order at most $k$
and that each vertex of $V(G) - T$ can be partitioned into at most $2^{k}$ types with respect to the adjacency to $T$.
This implies that the components of $V(G) - T$ can be classified into at most $4^{k^{2}}$ types 
($2^{k^{2}}$ for the isomorphism type and $(2^{k})^{k}$ for the adjacency to $T$)
in such a way that if two components $C_{1}$ and $C_{2}$ of $G-T$ are of the same type,
then there is an automorphism of $G$ that fixes $T$ and maps $C_{1}$ to $C_{2}$.
Given this classification of the components in $V(G) - T$,
we only need to guess how many components of each type are included in $\mathcal{C}$.
For this guess, we have at most 
$\binom{4^{k^{2}} + k^{2} - 1}{k^{2}} 
< 4^{k^{4} + k^{2}}$
options.

For each guess $\mathcal{C}$,
we guess the edges connecting the components in $\mathcal{C}$ to $T-R$ in $G-F$.
Since $|\mathcal{C}| < k^{2}$ and $|C| \le k$ for each $C \in \mathcal{C}$,
there are at most $k^{3} \cdot |T-R| \le k^{4}$ candidate edges.
We just try all $O(2^{k^{4}})$ subsets $F_{2}'$ of such edges,
and set $F_{2} = E(V(G)-T, T-R) - F_{2}'$.
In total, we have $O(2^{k^{4} + k^{2}} \cdot 2^{k^{4}})$ options for $F_{2}$.

\textbf{Guessing $F_{3}$:}
Recall that $G - T$ does not contain any component of order more than $k$.
Hence, if $G - R - (F_{1} \cup F_{2})$ has a component of order more than $k$, then it consists of 
some vertices in $T-R$ and some components in $\mathcal{C}$.
Thus, we only need to pick some edges of the components in $\mathcal{C}$ for $F_{3}$
to make $R$ a $\vi{k}$ set of $G-F$.
We use all $2^{k^{4}}$ subsets of the edges of the components in $\mathcal{C}$
as a candidate of $F_{3}$.

In total, $F = F_{1} \cup F_{2} \cup F_{3}$ has at most $2^{k^{2}} \cdot 4^{k^{4} + k^{2}} \cdot 2^{k^{4}} \cdot 2^{k^{4}}$ candidates,
and each candidate can be found in FPT time.
We reject this guess $F$ if $R$ is not a $\vi{k}$ set of $G-F$.
In the following, we assume that $F$ is guessed correctly and denote $G - F$ by $G'$.

\paragraph{Extending $\eta$:}
Recall that we already know how $\eta$ maps $S$ to $R$
and that each component in $Q - S$ and $G'-R$ is of order at most $k$.
We now extend $\eta$ by determining how $\eta$ maps $Q - S$ to $G'-R$.
By renaming vertices, we can assume that $S = \{s_{1}, \dots, s_{q}\}$,
$R = \{r_{1}, \dots, r_{q}\}$, and $\eta(s_{i}) = r_{i}$ for $1 \le i \le q$.

We say that a vertex $u$ in $Q-S$ \emph{matches} a vertex $v$ in $G' - R$
if $\{i \mid s_{i} \in N_{Q}(u) \cap S\} \subseteq \{i \mid r_{i} \in N_{G'}(v) \cap R\}$.
A set of components $\{C_{1}, \dots, C_{h}\}$ of $Q - S$ \emph{fits} a component $D$ of $G' - R$
if there is an isomorphism $\phi$ from the disjoint union of $C_{1}, \dots, C_{h}$ to $D$
such that for all $u \in \bigcup_{i} V(C_{i})$ and $v \in V(D)$, 
$\phi(u) = v$ holds only if $u$ matches $v$.
Note that if $h > k$, then $\{C_{1}, \dots, C_{h}\}$ can fit no component of $G'-R$.

As we did before for guessing $F_{2}$,
we classify the components of $Q-S$ and $G' - R$ into at most $4^{k^{2}}$ types.
Two components $C_{1}$ and $C_{2}$ of $Q-S$ (or of $G'-R$) are of the same type if and only if
there is an isomorphism $\phi$ from $C_{1}$ to $C_{2}$ such that
$\phi(v_{1}) = v_{2}$ implies that $N_{Q}(v_{1}) \cap S = N_{Q}(v_{2}) \cap S$
(or $N_{G'}(v_{1}) \cap R = N_{G'}(v_{2}) \cap R$, respectively).
We denote by $t(C)$ the type of a component $C$
and by $t(\{C_{1}, \dots, C_{h}\})$ the multi-set $\{t(C_{1}), \dots, t(C_{h}) \}$.
Observe that 
$\{C_{1}, \dots, C_{h}\}$ fits $D$
if and only if 
all sets $\{C_{1}', \dots, C_{h}'\}$ with $t(\{C_{1}', \dots, C_{h}'\}) = t(\{C_{1}, \dots, C_{h}\})$
fits $D'$ with $t(D') = t(D)$.

Observe that the guessed part $\eta|_{S}$ can be extended to
a subgraph isomorphism $\eta$ from $Q$ to $G'$ 
if and only if
there is a partition of the components of $Q - S$
such that each part $\{C_{1}, \dots, C_{h}\}$ in the partition can be injectively mapped to a component $D$ of $G'-R$
where $\{C_{1}, \dots, C_{h}\}$ fits $D$.
To check the existence of such a partition,
we only need to find for each pair of a multi-set $\mathcal{T}$ of types of a set of components in $Q-S$
and a type $\tau$ of a component in $G' - R$,
how many sets of components of type $\mathcal{T}$
the map $\eta$ embeds to components of type $\tau$.
We use the following ILP formulation to solve this problem.

Let $n_{\tau}$ and $n'_{\tau}$ be the numbers of type-$\tau$ components in $Q-S$ and $G'-R$, respectively.
These numbers can be computed in FPT time parameterized by $k$.

For each type $\tau$ and for each multi-set $\mathcal{T}$ of types such that $\mathcal{T}$ fits $\tau$,
we use a variable $x_{\mathcal{T}, \tau}$ to represent
the number of type-$\mathcal{T}$ multi-sets of components in $Q - S$
that are mapped to type-$\tau$ components in $G' - R$.
For each type $\tau$ of components in $G' - R$, 
we can embed at most $n_{\tau}$ sets of components in $Q - S$.
This constraint is expressed as follows:
\begin{align}
  n_{\tau} \ge \sum_{\mathcal{T}\colon \mathcal{T} \text{ fits } \tau}  x_{\mathcal{T}, \tau} \quad \text{for each type } \tau. 
  \label{eq:n-tau}
\end{align}
For each type $\sigma$ of components in $Q - S$,
we need to embed all $n_{\sigma}$ components of type $\sigma$
into some components of $G - R'$.
We can express this constraint as follows:
\begin{align}
  n_{\sigma} =
  \sum_{\mathcal{T}, \tau \colon \sigma \in \mathcal{T} \text{ and } \mathcal{T} \text{ fits } \tau}
  \mu_{\mathcal{T}, \sigma} \cdot x_{\mathcal{T}, \tau} \quad \text{for each type } \sigma,
  \label{eq:n-sigma}
\end{align}
where $\mu_{\mathcal{T}, \sigma}$ is the multiplicity of $\sigma$ in $\mathcal{T}$.
This completes the ILP formulation of the problem. 
We do not have any objective function and just ask for the feasibility.
The construction can be done in FPT time parameterized by $k$.

Observe that there are at most $\binom{4^{k^{2}} + k - 1}{k} < 4^{k^{3}+k}$ multi-sets $\mathcal{T}$ of types of components.
Thus the ILP above has at most $4^{k^{2}} \cdot 4^{k^{3}+k}$ variables (the first factor for $\tau$ and the second for $\mathcal{T}$)
and at most $4^{k^{2}} \cdot 4^{k^{3}+k} + 4^{k^{2}} \cdot 4^{k^{2}} \cdot 4^{k^{3}+k}$ constraints 
(the first term for \eqref{eq:n-tau} and the second for \eqref{eq:n-sigma})
of length $O(4^{k^{2}} \cdot 4^{k^{3}+k})$.
The coefficients are upper bounded by $|V(G')|$.
It is known that the feasibility check of such an ILP can be done in FPT time parameterized by
$k$~\cite{Lenstra83,Kannan87,FrankT87}. 
Thus, the problem can be solved in FPT time when parameterized by $k$.
\qed
\end{proof}

\subsection{Subgraph isomorphism on $k P_{4}$-minor free graphs}
We show that \textsc{Subgraph Isomorphism} on $k P_{4}$-minor free graphs is randomized XP-time solvable parameterized by $k$.
Our randomized algorithm is with false negatives. 
That is, it always rejects a \textsc{no}-instance, but may reject a \textsc{yes}-instance with probability at most $1/2$.

For a graph $G = (V,E)$, a set $S \subseteq V$ is a \emph{$P_{4}$-hitting set} of $G$
 if $G-S$ does not contain $P_{4}$ as a minor (or equivalently as a subgraph).
The \emph{$P_{4}$-hitting number} of $G$ is the minimum size of a $P_{4}$-hitting set of $G$.
To show the main result of this section, we prove Theorem~\ref{thm:P4-free-deletion} below, 
which immediately gives the result claimed above on $k P_{4}$-minor free graphs as their $P_{4}$-hitting number is at most $4(k-1)$.

Our algorithm will find a subgraph isomorphism $\eta$ from $Q$ to $G$ as follows:
\begin{enumerate}
  \item Find a $P_{4}$-hitting set $T$ of $G$ such that $|T| \le k$,
  and guess the ``used'' part $R$ of $T$.

  \item Guess $S = \eta^{-1}(R)$ and the mapping from $S$ to $R$.

  \item Color the vertices of $G - T$ and $Q - S$ according to the connections to $R$ and $S$, respectively.

  \item Guess how many vertices of each color $c$ in $G$ will remain unused
  after embedding the non-singleton components (triangles and stars) in $Q-S$ to $G-T$.
  Check whether the singleton components in $Q$ can be embedded to the guessed vertices.

  \item Construct an auxiliary bipartite multi-graph $B = (X,Y \cup Z;F)$ 
  from the components of $G$ and the non-singleton components in $Q$.
  
  \item Find a perfect matching of $B$ with a specific weight.
  Using such a matching, extend the guessed parts of $\eta$ to a subgraph isomorphism from $Q$ to $G$.
\end{enumerate}

\begin{theorem}
\label{thm:P4-free-deletion}
\textsc{Subgraph Isomorphism} on graphs with $P_{4}$-hitting number at most $k$
admits a randomized $n^{O(2^{k})}$-time algorithm with false negatives.
\end{theorem}

\begin{proof}
Let $G$ and $Q$ be graphs of $P_{4}$-hitting number at most $k$.
We will find a subgraph isomorphism $\eta$ from $Q$ to $G$ in randomized XP time parameterized by $k$.
In the following, we denote by $n$ and $m$
the total numbers of vertices and edges, respectively, in $G$ and $Q$.

\paragraph{Mapping $P_{4}$-hitting sets:}
We first find a $P_{4}$-hitting set $T$ of $G$ such that $|T| \le k$.
This can be done in time $O(4^{k} (n+m))$ by branching on $P_{4}$-subgraphs and checking $P_{4}$-subgraph freeness.
We guess $R = \eta(V(Q)) \cap T$, $S = \eta^{-1}(R)$, and $\eta(s) \in R$ for each $s \in S$.
We reject the guess if $\{u,v\} \in E(Q[S])$ and $\{\eta(u),\eta(v)\} \notin E(G[R])$ for some $u, v \in S$.
Also, if $S$ is not a $P_{4}$-hitting set of $Q$, then we reject this guess as we cannot map $Q-S$ to $G-T$.
We have $O(2^{k})$ options for $R$, $O(n^{k})$ options for $S$, and $O(k!)$ options for the mapping from $S$ to $R$.
By Lemma~\ref{lem:P4-minor-free_characterization}, each component of $G-T$ and $Q-S$ is
either a singleton $K_{1}$, a triangle $K_{3}$, or a star $K_{1,\ell}$ for some $\ell \ge 1$.

\paragraph{Coloring vertices:}
We rename the vertices in $R$ and $S$ to have $R = \{r_{1}, \dots, r_{q}\}$, $S = \{s_{1}, \dots, s_{q}\}$, 
and $\eta(s_{i}) = r_{i}$ for $1 \le i \le q$.
We set the \emph{color} of a vertex $v \in G - T$, denoted $\col(v)$, to be the set $\{i \mid r_{i} \in N_{G}(v)\}$.
Similarly, we set the color $\col(u)$ of a vertex $u \in Q - S$ to be $\{i \mid s_{i} \in N_{Q}(u)\}$.
Observe that $u \in Q - S$ can be embedded to $v \in G - T$ only if $\col(u) \subseteq \col(v)$ (assuming that the guesses so far are correct).

For a set of vertices $X$ of $G-T$ or $Q-S$ and a color $C \subseteq \{1,\dots,q\}$,
we set $h_{X}(C)$ to be the number of vertices $v \in X$ such that $\col(v) = C$.
We call $h_{X}$ the \emph{color histogram} of $X$.

In the later steps, it is convenient to identify color histograms with $2^{q}$-digit $n$-ary numbers.
Ordering the subsets of $\{1,\dots,q\}$ in an arbitrary way,
the $i$th digit can represent the number of vertices having the $i$th subset as their color.
Then for disjoint sets $X$ and $Y$, it holds that $h_{X} + h_{Y} = h_{X \cup Y}$.

\paragraph{Guessing the color histogram of unused vertices:}
We guess the color histogram $h_{A}$ of the set $A$ of vertices
that remain unused after embedding the non-singleton components (triangles and stars) in $Q-S$ to $G-T$.
(Note that we do not guess $A$ directly.)
The number of possible options for $h_{A}$ is $O(n^{2^{k}})$.
At this point, we do not know whether there is an embedding of the non-singleton components consistent with $h_{A}$.
For now, we assume the existence of such an embedding,
and first test whether the singleton components of $Q-S$ can be embedded to the unused vertices of $G-T$ guessed as $h_{A}$.

We need to embed each singleton component $u$ in $Q-S$ to a vertex $v$ such that $\col(u) \subseteq \col(v)$.
So the problem here can be reduced to the problem of finding a matching saturating $U$
in the bipartite graph such that
\begin{itemize}
  \item the vertex set is $U \cup V$;
  \item $U$ is the set of singleton components of $Q-S$;
  \item $V$ is a set of vertices that contains exactly $h_{A}(C)$ vertices of color $C$;
  \item $u \in U$ and $v \in V$ are adjacent if and only if $\col(u) \subseteq \col(v)$.
\end{itemize}
Since $|U \cup V| \le n$, we can check this in polynomial time.

\paragraph{Embedding non-singleton components:}
We finally test the existence of an embedding of the stars and triangles in $Q-S$ to $G-T$ consistent with $h_{A}$.
We reduce this task to the problem of deciding the existence of a perfect matching with a specific weight
in the bipartite multi-graph $B = (X, Y \cup Z; E)$ defined as follows:
\begin{enumerate}
  \item $X$ corresponds to the components of $G - T$.
  \item $Y$ corresponds to the non-singleton components in $Q-S$.
  \item $Z$ is the set of dummy vertices such that $|Z| = |X| - |Y|$. 
  (Observe that $|X| - |Y|$ has to be nonnegative if $S$ is guessed correctly.)
  
  \item \label{itm:nontrival-step}
  For $x \in X$ and $y \in Y$,
  add one edge of weight $h_{D}$ if there is a way to embed the component $C_{y}$ corresponding to $y$ to
  the component $C_{x}$ corresponding to $x$ such that the set of remaining vertices is $D$. 
  (There could be multiple edges with different weights between $x$ and $y$.)
  
  \item For $x \in X$ and $z \in Z$, add an edge of weight $h_{D}$, where
  $D$ is the set of vertices of the component $C_{x}$ corresponding to $x$.
\end{enumerate}
The graph $B$ can be constructed in time $n^{O(2^{k})}$.
To see this, the \ref{itm:nontrival-step}th step is the only nontrivial one.
For that step, we have $n^{2^{k}}$ candidates for $h_{D}$.
Each candidate can be checked in time polynomial in $n + 2^{k}$
since each component involved is $K_{1}$, $K_{3}$, or $K_{1,s}$.

From the construction, there exists an embedding of the non-singleton components in $Q-S$ to $G-T$ consistent with $h_{A}$
if and only if $B$ has a perfect matching of weight exactly $h_{A}$.
Including an edge between $x \in X$ and $y \in Y$ of weight $h_{D}$ into the perfect matching
means mapping $C_{y}$ to $C_{x}$ in such a way that $V(C_{x}) \setminus \eta(V(C_{y}))$ has the color histogram $h_{D}$.
Including an edge between $x \in X$ and $z \in Y$ means that $C_{x}$ is not used to embed any non-singleton component of $Q-S$.

Now it suffices to find a perfect matching of $B$ with weight exactly $h_{A}$.
It is known that, given a multi-graph with the maximum edge weight bounded by $W$ and a target total weight $T$,
there is a randomized algorithm with false negative that 
finds a perfect matching of weight exactly $T$ (if any) in time polynomial in $|V(B)| + |E(B)| +W$~\cite{MulmuleyVV87}
(see also \cite{Marx04,MarxP14,JansenM15}).
Since $W$ is $n^{O(2^{k})}$, the theorem holds.
\qed
\end{proof}

\subsection{Subgraph isomorphism on $3 P_{5}$-minor free graphs}

A \emph{double star} $D_{a,b}$ is the graph obtained from two stars $K_{1,a}$ and $K_{1,b}$
by connecting the centers of the stars with an edge.
\begin{theorem}
\label{thm:4P5}
\textsc{Subgraph Isomorphism} on $3 P_{5}$-minor free graphs is NP-complete. 
\end{theorem}
\begin{proof}
The problem is clearly in NP\@.
We show the NP-hardness by a reduction from 3-\textsc{Sat} with the restriction that
each clause contains two or three literals
and that each variable occurs exactly twice as a positive literal and exactly once as a negative literal.
We call this variant {\bsat}.
It is known that {\bsat} is NP-complete~\cite{FellowsKMP95}.
Let $(U, \mathcal{C})$ be an instance of {\bsat} with
the variables $U = \{u_{0}, \dots, u_{n-1}\}$ and the clauses $\mathcal{C} = \{C_{0}, \dots, C_{m-1}\}$.

We first construct the host graph $G$ (see \figref{fig:3p5mfreeG}).
It consists of $2n$ double stars, two special vertices $c$ and $c'$,
and many pendant vertices attached to $c$ and $c'$.
For each variable $u_{i} \in U$, we take two isomorphic double stars $D_{(n+i)n, (3n-i)n}$ and call them $D_{i}$ and $\overline{D}_{i}$.
The double stars $D_{i}$ and $\overline{D}_{i}$ correspond to the positive literal $u_{i}$
and the negative literal $\overline{u}_{i}$, respectively.
We add for each $u_{i} \in U$ some edges between leaves of the double stars $D_{i}$ and $\overline{D}_{j}$ and the special vertices $c$ and $c'$ as follows.
Let $u_{i} \in C_{j}, C_{k}$ and $\overline{u}_{i} \in C_{\ell}$.
We arbitrarily and bijectively assign the two stars $K_{1,(n+i)n}$ and $K_{1,(3n-i)n}$ in $D_{i}$ to $C_{j}$ and $C_{\ell}$. 
From the star in $D_{i}$ assigned to $C_{h}$ ($h \in \{j,\ell\}$),
we pick $n+h$ leaves and make them adjacent to $c$,
and then pick $3n-h$ leaves from the remaining and make them adjacent to $c$.
Similarly, in $\overline{D}_{i}$, we choose one side of the double star,
make $n+\ell$ of the leaves there adjacent to $c$,
and make $3n-\ell$ of the remaining adjacent to $c'$.
So far, we took $N := 4n^{3} + 2n + 2$ vertices into $G$.
By attaching $N$ and $2N$ pendant vertices to $c$ and  $c'$, respectively,
we complete the construction of $G$.
Observe that if we remove $c$ and $c'$ from $G$, then
it becomes a collection of double stars and isolated vertices, which is $P_{5}$-minor free.
Thus the host graph $G$ cannot have $3 P_{5}$ as a minor.

\begin{figure}[tb]
  \centering
  \includegraphics[scale=.6]{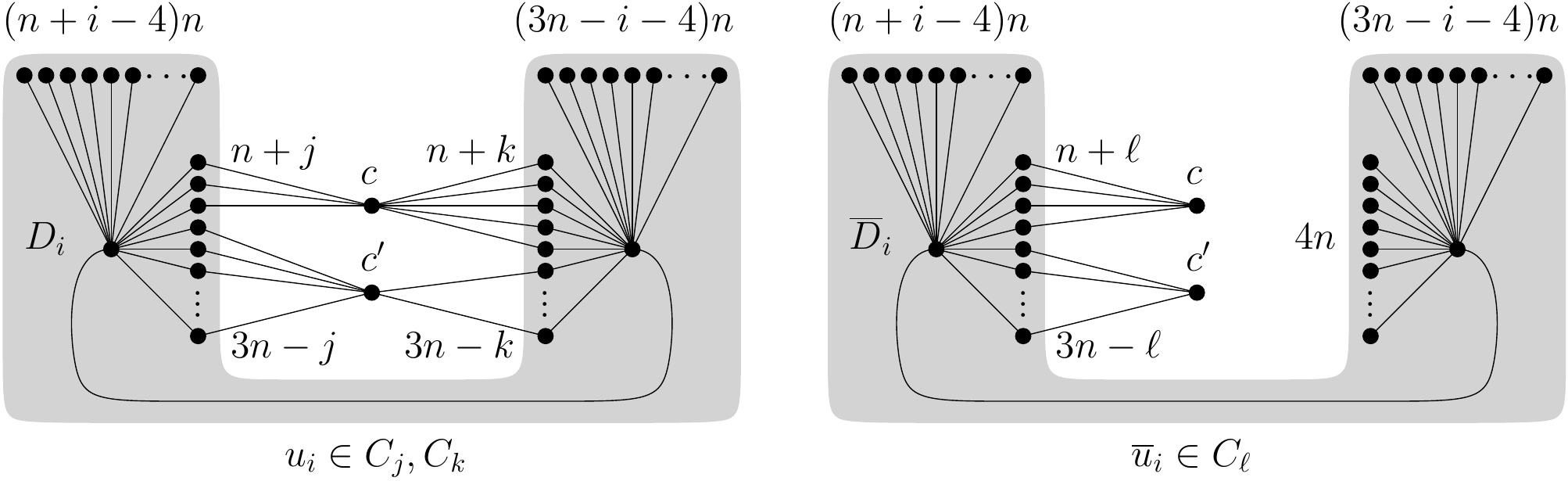}
  \caption{The gadgets in the host graph $G$. Note that the special vertices $c$ and $c'$ are shared by all gadgets
  although their copies are depicted for each gadget for the readability. 
  We omit the pendant vertices attached to $c$ and $c'$.}
  \label{fig:3p5mfreeG}
\end{figure}

The pattern graph $Q$ consists of $n$ double stars,
$m$ stars $K_{1, 4n}$, and two additional vertices $d$ and $d'$
(see \figref{fig:3p5mfreeQ}).
For each variable $u_{i} \in U$, we take a double star $D_{(n+i)n, (3n-i)n}$ and call it $B_{i}$.
For each clause $C_{j} \in \mathcal{C}$, we take a star $K_{1,4n}$ and call it $F_{j}$.
For each star $F_{j}$,
we arbitrarily select $n+j$ leaves and make them adjacent to $d$,
and then make the remaining $3n-j$ leaves adjacent to $d'$.
Finally, we attach $N$ and $2N$ pendant vertices to $d$ and  $d'$, respectively,
Observe that if we remove $d$ and $d'$, then
the the pattern graph becomes a collection of stars, double stars, and isolated vertices.
Thus $H$ is $3 P_{5}$-minor free.

\begin{figure}[tb]
  \centering
  \includegraphics[scale=.6]{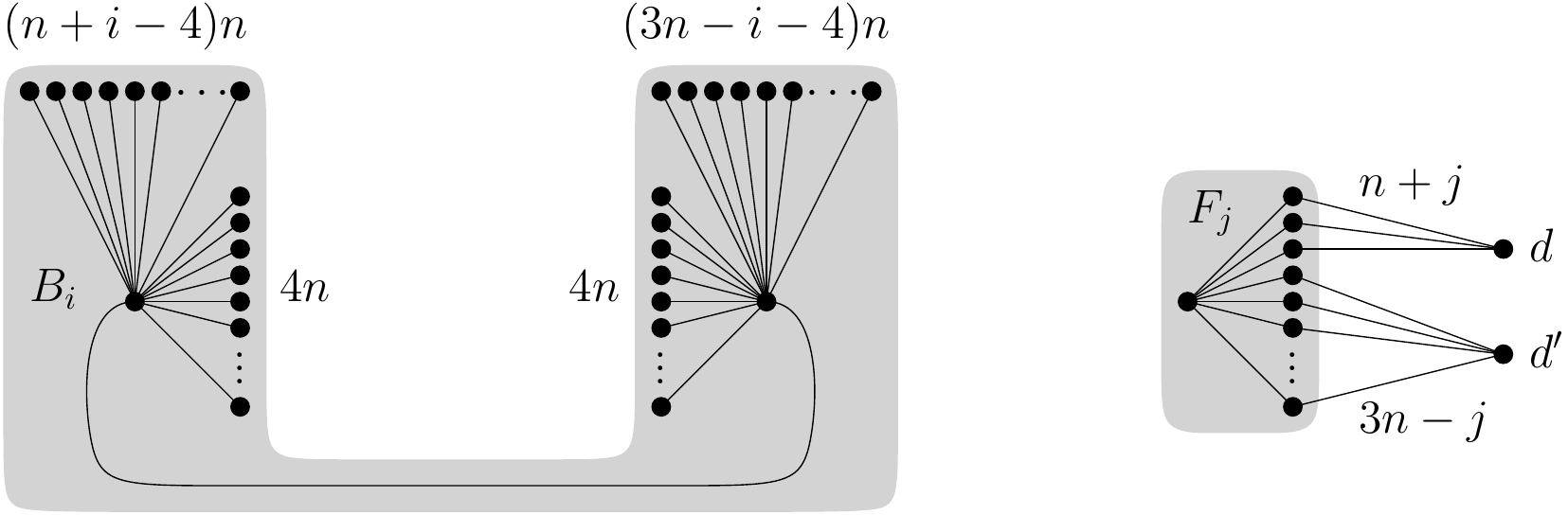}
  \caption{The gadgets in $Q$. We omit the pendant vertices attached to $d$ and $d'$.}
  \label{fig:3p5mfreeQ}
\end{figure}

In the following, we show that $Q \preceq G$
if and only if 
$(U, \mathcal{C})$ is a yes-instance of {\bsat}.

\paragraph{The if direction:}
Suppose that $(U, \mathcal{C})$ has a satisfying assignment $f \colon U \to \{\texttt{true},$ $\texttt{false}\}$.
We construct a subgraph isomorphism $\eta \colon V(Q) \to V(G)$ as follows.
We first set $\eta(d) = c$ and $\eta(d') = c'$.
We map the pendants attached to them appropriately.
For each $u_{i} \in U$, we map $B_{i}$ to $\overline{D}_{i}$ if $f(u_{i}) = \texttt{true}$,
and to $D_{i}$ if $f(u_{i}) = \texttt{false}$.

Now observe that what we left unused in $G$ for each $u_{i}$ is the double star
that corresponds to the clauses satisfied by the literal of $u_{i}$ that is \texttt{true} under $f$.
Furthermore, since $f$ is a satisfying assignment of $(U, \mathcal{C})$,
each clause $C_{j} \in \mathcal{C}$ has at least one star assigned to $C_{j}$
included in an unused double star in $G$. Also, such a correspondence is injective by the construction.
Therefore, we can map each $F_{j}$ in $Q$ into a corresponding star included in an unused double star in $G$
by mapping the center to the center, the neighbors of $d$ to the neighbors of $c$,
and the neighbors of $d'$ to the neighbors of $c'$.

\paragraph{The only if direction:}
Suppose that $Q \preceq G$ and thus there is a subgraph isomorphism $\eta \colon V(Q) \to V(G)$ from $Q$ to $G$.
Observe that $\eta(d) = c$ and $\eta(d') = c'$ because of their high degrees.
Recall that each double star in $Q$ and $G$ has exactly $4n^{2} + 2$ vertices
and that $B_{i}$ is isomorphic only to $D_{i}$ and $\overline{D}_{i}$.
Hence, $B_{i}$ has to be mapped to $D_{i}$ or $\overline{D}_{i}$ for every $i$.

Now we know that after the vertices $\{d,d'\}\cup \bigcup_{i} V(B_{i})$ are mapped,
the unused vertices in $G$ induce a subgraph that contains either $D_{i}$ or $\overline{D}_{i}$ for each $i$.
From this subgraph, we construct a truth assignment $f \colon U \to \{\texttt{true}, \texttt{false}\}$:
if $D_{i}$ is left unused we set $f(u_{i}) = \texttt{true}$;
otherwise we set $f(u_{i}) = \texttt{false}$.

Assume that $D_{i}$ is left unused by $\{d,d'\}\cup \bigcup_{i} V(B_{i})$
and $\eta(V(F_{j})) \subseteq V(D_{i})$ holds for some $i$.
(The other case where $D_{i}$ is replaced with $\overline{D}_{i}$ is the same.)
Since each leaf in $F_{j}$ is adjacent to either $d$ or $d'$, 
each image of them has to be adjacent to either $c$ or $c'$.
Thus $F_{j}$ is mapped to one of the stars in $D_{i}$.
Because of the connections to $d$ and $d'$ in $Q$ and to $c$ and $c'$ in $G$,
this is possible only if the star induced by $\eta(V(F_{j}))$ corresponds to the clause $C_{j}$.
Thus $D_{i}$ contains a star corresponding to $C_{j}$,
and $C_{j}$ includes the positive literal of $u_{i}$.
Since $D_{i}$ is left unused, we have $f(u_{i}) = \texttt{true}$, and thus $C_{j}$ is satisfied by $f$.
\qed
\end{proof}


\subsection{Missing cases}
\label{ssec:missing-cases}

To complete the proof of Theorem~\ref{thm:disconnected}, here we discuss the missing cases of \textsc{Subgraph Isomorphism} on $H$-minor free graphs.
In the missing cases, $H$ contains, as a minor, $P_{5}$ or $2P_{5}$ but no $3P_{5}$ nor $P_{6}$.
In other words, there is $p \in \{1,2\}$ such that $H$ contains a $p P_{5}$-minor but no $(p+1)P_{5}$-minor nor $P_{6}$-minor.

We reduce this case to the case where the forbidden minor is $p P_{5}$ in randomized polynomial time.
Since $H$ is a linear forest, $H$ is a minor of $p P_{5} \cup k P_{4}$ for some constant $k < |V(H)|$.
An $H$-minor free graph is $(p P_{5} \cup k P_{4})$-minor free, and thus it is either $p P_{5}$-minor free or $(k+5p) P_{4}$-minor free.
(The idea here is basically the same with the one in the proof of Corollary~\ref{cor:kP3}.)
If the host graph $G$ is $p P_{5}$-minor free (or $(k+5p) P_{4}$-minor free) but the pattern graph $Q$ is not,
then we output a trivial \textsc{no}-instance (e.g., $G = K_{1}$ and $Q = K_{2}$).
If both $G$ and $Q$ are $p P_{5}$-minor free,
we just output the original input $G$ and $Q$ as the reduced $p P_{5}$-minor free instance.
If both $G$ and $Q$ are $(k+5p) P_{4}$-minor free,
then by Theorem~\ref{thm:P4-free-deletion}, we can solve the problem in randomized polynomial-time.
Depending on whether $Q \preceq G$ or not, we output a trivial \textsc{yes}-instance (e.g., $G = Q = K_{1}$) or a trivial \textsc{no}-instance.


\section{Structural parameterizations of \textsc{Subgraph Isomorphism}}
\label{sec:structural-parameterization}

We conclude the paper with some remarks on structural parameterizations of \textsc{Subgraph Isomorphism}.
Our results imply a few things in this direction. See \figref{fig:width-parameters}.
The proof of Theorem~\ref{thm:p6} implies that \textsc{Subgraph Isomorphism}
is NP-complete even for graphs of tree-depth~\cite{NesetrilDM12} at most 3.
This bound is tight by Lemma~\ref{lem:p4} since graphs of tree-depth at most 2 does not contain $P_{4}$ as a subgraph.
Proposition~\ref{prop:linfor-cluster} implies it is NP-complete even for graphs of constant twin-cover number~\cite{Ganian15}
because cluster graphs have twin-cover number 0.
For the parameterization by neighborhood diversity~\cite{Lampis12},
we can use techniques similar to the ones we used for Theorem~\ref{thm:vertex-integrity}.

\begin{figure}[tb]
  \centering
  \includegraphics[scale=.8]{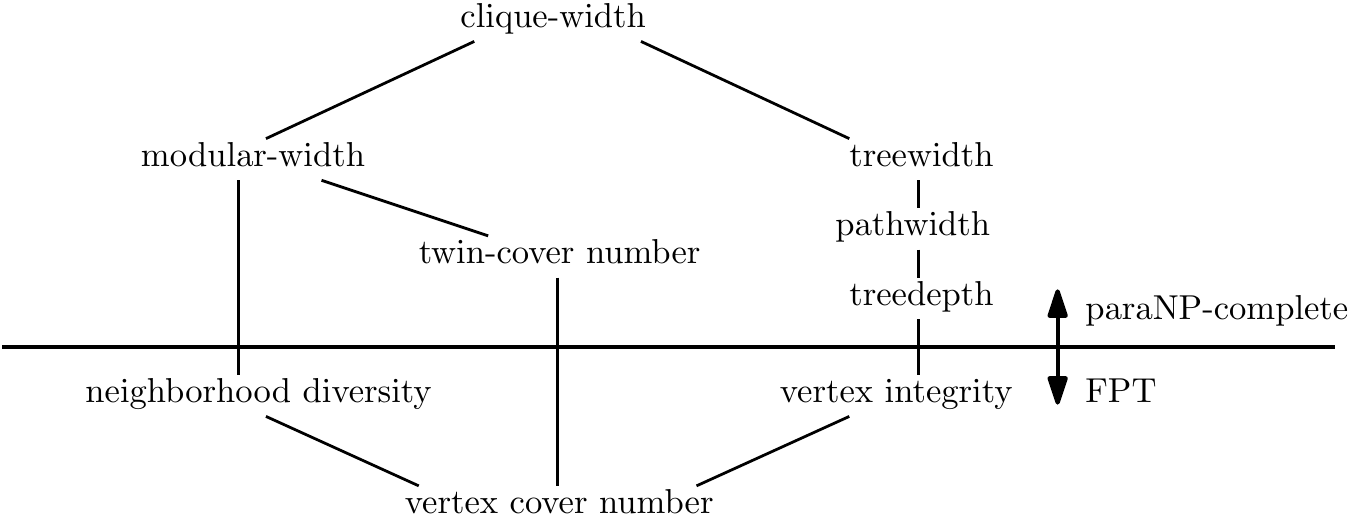}
  \caption{Graph parameters and \textsc{Subgraph Isomorphism}.
  For each connection of parameters, there is a function in the parameter above that lower bounds the one below.}
  \label{fig:width-parameters}
\end{figure}

Two vertices $u$ and $v$ of a graph $G = (V,E)$ are \emph{twins} if $N(u) \setminus \{v\} = N(v) \setminus \{u\}$.
The \emph{neighborhood diversity} of $G = (V,E)$ is the minimum integer $k$ such that
$V$ can be partitioned into $k$ sets $T_{1}, \dots, T_{k}$ of pairwise twin vertices.
Such a minimum partition can be found in linear time using fast modular decomposition algorithms~\cite{McConnellS99,TedderCHP08}.
Observe that each part $T_{i}$ in the partition is either complete or independent.
Also, for parts $T_{i}$ and $T_{j}$, there are either no edges or all possible edges.
We say that $T_{i}$ and $T_{j}$ are \emph{adjacent} if there are all possible edges,
and \emph{nonadjacent} otherwise.

\begin{theorem}
\label{thm:nd}
\textsc{Subgraph Isomorphism} on graphs of neighborhood diversity at most $k$ is fixed-parameter tractable parameterized by $k$.
\end{theorem}
\begin{proof}
Let $G$ be the host graph and $Q$ be the pattern graph, both with neighborhood diversity at most $k$.
Let $T_{1}, \dots, T_{t}$ be a partition of the vertices of $G$ into pairwise twin vertices with $t \le k$,
and similarly let $R_{1}, \dots, R_{r}$ be a partition of the vertices of $Q$ into pairwise twin vertices with $r \le k$.

We find a subgraph isomorphism $\eta$ from $Q$ to $G$ by reducing the problem to at most $3^{k^{2}}$ instances of ILP as follows.
By a variable $x_{i,j}$ we represent the number of the vertices that $\eta$ maps from $R_{i}$ to $T_{j}$.
For each variable $x_{i,j}$, we guess whether $x_{i,j} = 0$, $x_{i,j} = 1$, or $x_{i,j} \ge 2$.
Since we have at most $k^{2}$ variables $x_{i,j}$, this gives us at most $3^{k^{2}}$ options.

We reject this guess if at least one of the following holds:
\begin{itemize}
  \item $x_{i,j} \ne 0$ and $x_{i',j'} \ne 0$ for adjacent $R_{i}$ and $R_{i'}$ and nonadjacent $T_{j}$ and $T_{j'}$;
  
  \item $x_{i,j} \ne 0$ and $x_{i',j} \ne 0$ for adjacent $R_{i}$ and $R_{i'}$ and an independent $T_{j}$;

  \item $x_{i,j} \ge 2$ for a complete $R_{i}$ and an independent $T_{j}$.
\end{itemize}

For guesses satisfying all the conditions above, we construct an ILP instance as follows.
For each variable $x_{i,j}$, we add the guessed constraint $x_{i,j} = 0$, $x_{i,j} = 1$, or $x_{i,j} \ge 2$.
For each $i \in \{1, \dots, r\}$, we add the constraint $\sum_{1 \le j \le t} x_{i,j} = |R_{i}|$.
For each $j \in \{1, \dots, t\}$, we add the constraint $\sum_{1 \le i \le r} x_{i,j} \le |T_{j}|$.
We can see that this ILP is feasible if and only if there is a subgraph isomorphism consistent with the guess.
As we saw in the proof of Theorem~\ref{thm:vertex-integrity} for vertex integrity, this feasibility check can be done in FPT time parameterized by $k$.
\qed
\end{proof}


\bibliographystyle{plainurl}
\bibliography{sgimf}

\end{document}